%% file: Full-PaperACA2021_v14.tex
\documentclass{svproc}
\usepackage{url}

\usepackage[all]{xy}
\usepackage[mathscr]{eucal}
\include{macrosB}
\include{newpack3}

\usepackage{mathptmx}       % selects Times Roman as basic font
\usepackage{helvet}         % selects Helvetica as sans-serif font
\usepackage{courier}        % selects Courier as typewriter font
\usepackage{type1cm}        % activate if the above 3 fonts are
\usepackage{tikz}
\usepackage{tikz-cd}              % not available on your system
\usepackage{makeidx}         % allows index generation
\usepackage{graphicx}        % standard LaTeX graphics tool
\usepackage{multicol}        % used for the two-column index
\usepackage[bottom]{footmisc}% places footnotes at page bottom

\makeindex             % used for the subject index
\usepackage[most]{tcolorbox}
\usetikzlibrary{patterns}
\pgfdeclarepatternformonly{mystrikeout}{\pgfqpoint{-1pt}{-1pt}}
{\pgfqpoint{11pt}{11pt}}{\pgfqpoint{10pt}{10pt}}%
{
  \pgfsetlinewidth{0.4pt}
  \pgfpathmoveto{\pgfqpoint{0pt}{0pt}}
  \pgfpathlineto{\pgfqpoint{10.1pt}{10.1pt}}
  \pgfusepath{stroke}
}
\newtcolorbox{tcbstrikeout}{breakable,
 enhanced jigsaw,
 opacityback=0,
 parbox=false,
 boxrule=0mm,
 top=0mm,bottom=0pt,left=0pt,right=0pt,
 boxsep=0pt,
 frame hidden,
 finish={\fill[pattern=mystrikeout] (frame.north west) rectangle (frame.south east);}
}
\def\calB{{\mathcal B}}
\def\D{{\mathcal D}}

\def\calH{{\mathcal H}}

\def\calX{{\mathcal X}}

\newcommand{\kk}{\mathbf{k}}
\def\Dom{\mathrm{Dom}}

\def\QY{\Q\langle Y \rangle}

\def\scal#1#2{\langle #1\bv#2 \rangle}
\def\bv{\mid}

\def\abs#1{\bv\!#1\!\bv}

\def\der{\mathbf{d}}
\def\ul#1{\underline{#1}}
\newcounter{per1}
\begingroup
\def\2#1{\ifnum#1<10 0\fi\the#1}
\xdef\isodayandtime{
{\2\day-\2\month-\the\year\space\2{\count0}:%
\2{\count2}}}
\endgroup
 \usepackage{xcolor}
\begin{document}
\mainmatter              % start of a contribution
\title{Towards a Theory of Domains for Harmonic Functions and its Symbolic Counterpart.}
\titlerunning{Theory of Domains \& its Symbolic Counterpart.}  % abbreviated title (for running head)
\author{Bui Van Chien\inst{4}, G\'erard H. E. Duchamp\inst{1},
Ngo Quoc Hoan\inst{3},\\ V. Hoang Ngoc Minh\inst{2}, Nguyen Dinh Vu\inst{5}    }
\authorrunning{B. V. Chien, G. Duchamp, N. Q. Hoan, H.N. Minh, N. D. Vu} % abbreviated author list (for running head)
\tocauthor{Bui Van Chien, G\'erard H. E. Duchamp,
Ngo Quoc Hoan, Hoang Ngoc Minh Vincel, Nguyen Dinh Vu  }
\institute{LIPN - UMR 7030, CNRS, 93430 Villetaneuse, France,\\
\email{gheduchamp@gmail.com},\\ 
\and
University of Lille, 1, Place D\'eliot, 59024 Lille, France,\\
\email{ vincel.hoang-ngoc-minh@univ-lille.fr}
\and
Hai Phong University, 171, Phan Dang Luu, Kien An, Hai Phong, Viet Nam \\
\email{quochoan\_ngo@yahoo.com.vn}
\and
Hue University, 77, Nguyen Hue, Hue, Viet Nam\\
 \email{bvchien.vn@gmail.com} 
\and  Institute of Mathematics, Vietnam Academy of Science and Technology,\\ 18 Hoang Quoc Viet, 10307 Hanoi, Vietnam\\
\email{ndvu@math.ac.vn}}

\maketitle              % typeset the title of the contribution

\begin{abstract}
In this paper, we begin by reviewing the calculus induced by the framework of \cite{GHM1}. In there, we extended Polylogarithm functions over a subalgebra of noncommutative rational power series, recognizable by finite state (multiplicity)
automata over the alphabet
$X=\{x_0,x_1\}$. The stability of this calculus under shuffle products relies on the nuclearity of the target space  \cite{Sch}. We also concentrated on algebraic and analytic aspects
of this extension allowing to index polylogarithms,
at non positive multi-indices, by rational series
and also allowing to regularize divergent polyzetas, at non positive multi-indices \cite{GHM1}. As a continuation of  works in \cite{GHM1} and in order to understand the bridge between the extension of this ``polylogarithmic calculus'' and the world of harmonic sums, we propose a local theory, adapted to a full calculus on indices of Harmonic Sums based
on the Taylor expansions, around zero, of polylogarithms with index $x_1$ on the rightmost end.  
This theory is not only compatible with Stuffle products but also with the Analytic Model. In this respect, it provides a stable and fully algorithmic model for Harmonic calculus. Examples by computer are also provided\footnote{This research of Ngo Quoc Hoan is funded by the Vietnam National
Foundation for Science and Technology Development (NAFOSTED) under grant number  101.04-2021.41}. 
% We would like to encourage you to list your keywords within
% the abstract section using the \keywords{...} command.
\end{abstract}
\section{Introduction}
Riemann's zeta function is defined by the series 
\begin{eqnarray}\label{Zetavalue}
\zeta (s)  &:=& \sum\limits_{n \geq 1} \dfrac{1}{n^s} 
\end{eqnarray}
where $s$ is a complex number. It is absolutely 
for $\Re (s) > 1$ (for any $s \in \C$, $\Re (s)$ stands for the real part of $s$).

It can be extended to a meromorphic function on the complex plane $\C$ 
with a single pole at $s = 1$ \cite{riemann}\footnote{Whence the famous sum 
$
\zeta(-1)=1+2+3+\cdots =-\frac{1}{12}
$
by which, among other ``results'', S. Ramanujan was noticed by G. H. H. Hardy (see \cite{SR1}).}.). In fact, the story began with Euler's works to find the solution of Basel's problem. In these works,  Euler proved  that \cite{euler1}
 \begin{eqnarray}
 \zeta (2 ) = \sum\limits_{n \geq 1} \dfrac{1}{n^2} = \dfrac{\pi^2}{6}.
\end{eqnarray} 
Moreover, for any $s_1, s_2 \in \C$ such that $\Re (s_1) > 1$ and $\Re (s_2) > 1$, Euler gave an important identity as follows\footnote{In fact, in Euler's formula, $s_1, s_2 \in \N_{+}$. This identity 
appeared under the name ``Prima Methodus ...'' (see \cite{euler2} pp 141-144).}:  
\begin{eqnarray}
\zeta (s_1) \zeta (s_2) &=& \zeta (s_1, s_2) + \zeta (s_1+s_2) + \zeta (s_2,s_1).
\end{eqnarray} 
where, for any $s_1, s_2 \in \C$ such that $\Re (s_1) > 1$ and $\Re (s_2) > 1$, 
\begin{eqnarray}
\zeta (s_1,  s_2) &:=& \sum\limits_{n_1 > n_2 \geq 1} \dfrac{1}{n_1^{s_1} n_2^{s_2}}. 
\end{eqnarray}
The numbers $\zeta (s_1, s_2)$ were called ``double zeta values'' at $(s_1, s_2)$. More generally, for any $r \in \N_+$ and $s_1, \ldots , s_r \in \C$, we denote 
\begin{eqnarray}
\zeta (s_1, \ldots , s_r) &:=& \sum\limits_{n_1 > \ldots > n_r \geq 1} \dfrac{1}{n_1^{s_1} \ldots n_r^{s_r}}. 
\end{eqnarray} 
Then the results of K. Matsumoto \cite{FKMT} showed that the series $\zeta (s_1, \ldots , s_r)$ converges absolutely for $s\in \calH_r $ where 
\begin{eqnarray}
\calH_r  &:=& \left\lbrace {\bf s} = (s_1, \ldots , s_r) \in \C^r | \forall m =1, \ldots , r; \ \Re (s_1) + \ldots + \Re(s_m) > m     \right\rbrace .
\end{eqnarray} 
In the convergent cases, $\zeta (s_1, \ldots , s_r)$  were called ``polyzeta values'' at multi-index ${\bf s} = (s_1, \ldots , s_r)$. 
Indeed $\mathbf{s}\mapsto \zeta(\mathbf{s})$ is holomorphic on $\calH_r$ and has been extended to $\C^r$ as a meromorphic function (see \cite{Goncharov,Zhao}). 

In fact, for any $r$-uplet $(s_1,\ldots,s_r)\in\N_+^r$, $r\in\N_+$, the  polyzeta $\zeta (s_1, \ldots , s_r)$ is also the limit at $z=1$ of the {\it polylogarithmic function}, defined by:
\begin{eqnarray}\label{polylogarithm}
\Li_{s_1,\ldots,s_r}(z):=\sum_{n_1>\ldots>n_r>0}\frac{z^{n_1}}{n_1^{s_1}\ldots n_r^{s_r}}
\end{eqnarray}
 for any $z\in\C$ such that $\abs{z}<1$. It is easily seen that, for any $s_i \in \N_+,\ r > 1$, 
 \begin{eqnarray}\label{recursion}
 z\dfrac{d}{dz} \Li_{s_1, \ldots , s_r} (z)&=& \Li_{s_1-1, \ldots , s_r} (z) \mbox{ if } s_1> 1\\
 (1 - z) \dfrac{d}{dz} \Li_{1, s_2, \ldots , s_r} (z) &=& \Li_{s_2, \ldots , s_r} (z) \mbox{ if } r > 1 
 \end{eqnarray}
 and this formulas will be ended at the ``seed'' $ \Li_1 (z) = \log \left(\dfrac{1}{1-z} \right).  $
 
  Moreover, if $X^*$ is the free monoid of rank two (generators, or alphabet, $X = \left\lbrace x_0,  x_1 \right\rbrace$ and neutral $1_{X^*})$ then the polylogarithms indexed by a list 
\begin{eqnarray}\label{re-index}
(s_1,\ldots,s_r)\in\N_{+}^r\mbox{ can be reindexed by the word }x_0^{s_1-1}x_1\ldots x_0^{s_r-1}x_1\in X^*x_1
\end{eqnarray}  
In order to reverse the recursion introduced in Eqns. \ref{recursion}, we introduce two differential forms 
\begin{eqnarray}
\omega_0(z)=z^{-1}dz&\mbox{and}&\omega_1(z)=(1-z)^{-1}dz,
\end{eqnarray} 
on $\Omega$\footnote{$\Omega$ is the simply connected domain $\C\setminus(]-\infty,0]\cup[1,+\infty[)$.}. 
We then get an integral representation\footnote{In here, we code the moves $z\dfrac{d}{dz}$ (resp. $(1-z) \dfrac{d}{dz}$) - or more precisely sections $\int_{0}^z \dfrac{f(s)}{s} ds$ (resp. $\int_0^z \dfrac{f(s)}{1-s} ds $) - with $x_0$ (resp. $x_1$).} of the functions \eqref{polylogarithm} as follows\footnote{Given
a word $w\in X^*$,  we note $|w|_{x_1}$ the number of occurrences of $x_1$ within $w$.} \cite{Minh1}
\begin{eqnarray}\label{Improper_Int}
\Li_w(z)=\left\{
\begin{array}{lclcl}
1_{\calH(\Omega)}&\mbox{if}&w=1_{X^*}\\[3mm]
\displaystyle\int_{0}^z\omega_1(s)\Li_u(s)&\mbox{if}& w=x_1u \\[3mm]
\displaystyle\int_{1}^z\omega_0(s)\Li_u(s)&\mbox{if}&w=x_0u&\mbox{and}&|u|_{x_1}=0, i.e. w\in x_0^*\\[3mm]
\displaystyle\int_{0}^z\omega_0(s)\Li_u(s)&\mbox{if}&w=x_0u&\mbox{and}&|u|_{x_1}>0, i.e. w\notin x_0^*,
\end{array}\right.
\end{eqnarray}
the upper bound $z$ belongs to $\Omega$ (we recall that 
$\Omega=\C\setminus(]-\infty,0]\cup[1,+\infty[)$
is simply connected domain so that the intergrals, which can be proved to be convergent in all cases, depend 
only on  their bounds). The the neutral element of the algebra of analytic functions $\calH(\Omega)$, a 
constant function will be here denoted 
$1_{\calH(\Omega)}$.\\
This provides not only the analytic
continuation of \eqref{polylogarithm} to $\Omega$ but also extends the indexation
to the whole alphabet $X$, allowing to study the complete generating series
\begin{eqnarray}\label{polylog2}
\L(z)=\sum_{w\in X^*}\Li_w(z)w
\end{eqnarray}
and show that it is the solution of the following first order noncommutative differential equation
\begin{eqnarray}\label{DrinfeldSys}
\left\{
\begin{array}{lcl}
\der(S)=(\omega_0(z)x_0+\omega_1(z)x_1)S,&&(NCDE)\cr
\lim\limits_{z\in\Omega,z\to 0}S(z)e^{-x_0\log(z)}=1_{\ncs{\calH(\Omega)}{X}},
&&\mbox{asymptotic initinial condition,}
\end{array}\right.
\end{eqnarray}
where, for any $S\in\ncs{\calH(\Omega)}{X}$.\\ 
Through term by term derivation, one gets \cite{Drin}
\begin{eqnarray}
\der(S)=\sum_{w\in X^*}\dfrac{d}{dz}(\scal{S}{w})w.
\end{eqnarray}
This differential system allows to show that $\L$ is a $\shuffle$-character\footnote{Here, the shuffle product is denoted by $\shuffle$.  It will be redefined in the section \ref{stuffledef} .}  \cite{MJOP}, \textit{i.e.}
\begin{eqnarray}\label{Li_shuffle_mor0}
\forall u,v\in X^*,\quad\scal{\L}{u\shuffle v}=\scal{\L}{u}\scal{\L}{v}
&\mbox{and}&\scal{\L}{1_{X^*}}=1_{\calH(\Omega)}.
\end{eqnarray}

Note that, in what precedes, we used the pairing  $\scal{\bullet}{\bullet}$ between series and polynomials,
classically defined by, for $T\in\ncs{\kk}{X}$ and $P\in\ncp{\kk}{X}$
%\footnote{Here $R$ is any commutative ring(like $\calH(\Omega),\C,\calZ[\gamma]$, ...).}  
\begin{eqnarray}
\scal{T}{P}=\sum_{w\in X^*}\scal{T}{w}\scal{P}{w}, 
\end{eqnarray}
where, when $w$ is a word, $\scal{S}{w}$ stands for the coefficient of $w$ in $S$ and $\kk$ any commutative ring (as here $\calH(\Omega)$). With this at hand,
we extend  at once the indexation of $\Li$ from $X^*$ to $\ncp{\C}{X}$ by 
\begin{eqnarray}\label{linext}
\Li_{P}:=\sum_{w\in X^*}\scal{P}{w}\Li_w=\sum_{n\ge0}\biggl(\sum_{|w|=n}\scal{P}{w}\Li_w\biggr).
\end{eqnarray}

In \cite{GHM1}, it has been established that the polylogarithm, well defined locally by \eqref{polylogarithm},
could be extended to some series (with conditions) by the last part of formula \eqref{linext} where the polynomial
$P$ is replaced by some series. A complete theory of global domains was presented in \cite{GHM1}, the present work concerns the whole project of extending 
$\H_{\bullet}$ \cite{GHM,Hoffman2}.\\
over stuffle subalgebras of rational power series on the alphabet $Y$, in particular
the stars of letters and some explicit combinatorial consequences of this extension.

In fact, we focus on what happens in (well choosen) neighbourhoods of zero (see section \ref{removable}), therefore, the aim of this work is manyfold. 

a) Use the extension to local Taylor expansions\footnote{Around zero.} as in \eqref{polylogarithm}
and the coefficients of their quotients by $1-z$, namely the harmonic sums, denoted $\H_\bullet$ and defined,
for any $w\in X^*x_1$, as follows\footnote{Here, the ${\tt conc}$-morphism
$\pi_X:(\ncp{\C}{Y},{\tt conc},1_{Y^*})\rightarrow(\ncp{\C}{X},{\tt conc},1_{X^*})$
is defined by $\pi_X(y_n)=x_0^{n-1}x_1$ and $\pi_Y$ is its inverse on $\mathrm{Im}(\pi_X)$.
See \cite{CHM1,GHM1} for more details and a full Definition of $\pi_Y$.} (\cite{Minh4} see also related literature \cite{BCK,Hoffman2})
\begin{eqnarray}
\frac{\Li_w(z)}{1-z}=\sum_{N\ge0}\H_{\pi_X(w)}(N)z^N,
\end{eqnarray}
by a suitable theory of local domains which assures to carry over the computation of these Taylor coefficients
and preserves the stuffle indentity, again true for polynomials over the alphabet $Y=\{y_n\}_{n\ge1}$, \textit{i.e.}\footnote{Here,  $\stuffle$ stands for the stuffle product which  will be recalled as in the section \ref{stuffledef}.}
\begin{eqnarray}
\forall S,T\in\ncp{\C}{Y},&&\H_{S\stuffle T}=\H_{S}\H_{T}
\mbox{ and }\H_{1_{\ncp{\C}{Y}}}=1_{\C^\N},
\end{eqnarray}
note that $1_{\ncp{\C}{Y}}$ is identified with $1_{Y^*}$ and $1_{\C^\N}$ is the constant (to one) function\footnote{In fact, it could be $\Q$ but we will use afterwards $\C$-linear combinations.} $\N\to \C$. This means that $\H_{\bullet}:(\ncp{\C}{Y},\stuffle,1_{Y^*})\longrightarrow(\C\{\H_w\}_{w\in Y^*},\times,1)$ mapping any word\\ $w=y_{s_1}\ldots y_{s_r}\in Y^*$ to
\begin{eqnarray}\label{HarmonicSums}
\H_w=\H_{s_1,\ldots,s_r}=\sum_{N\ge n_1>\ldots>n_r>0}\frac1{n_1^{s_1}\ldots n_r^{s_r}},
\end{eqnarray}
is a $\stuffle$ (unital) morphism\footnote{It can be proved that this morphism is into \cite{Minh4}.}.

b) Extend these correspondences (\textit{i.e.} $\Li_\bullet,\H_\bullet$) to some series (over $X$ and $Y$, respectively)
in order to preserve the identity\footnote{Here $\odot$ stands for the Hadamard product \cite{Had}.} \cite{Minh4}
\begin{eqnarray}\label{idP}
\frac{\Li_{\pi_X(S)}(z)}{1-z}\odot\frac{\Li_{\pi_X(T)}(z)}{1-z}=\frac{\Li_{\pi_X(S\stuffle T)}(z)}{1-z}.
\end{eqnarray}
true for polynomials $S,T\in\ncp{\C}{Y}$.

To this end, we use the explicit parametrization of the $\tt conc$-characters obtained in \cite{CHM1,GHM1} and the fact that,
under stuffle products, they form a group. 
\section{Polylogarithms: from global to local domains}
Now we are facing the following constraint: 

\textit{In order that the results given by symbolic computation reflect the reality with complex numbers (and analytic functions), we have to introduce some topology \footnote{Readers who are not keen on topology or functional analysis may skip the details of this section and hold its conclusions.}.}

\noindent  
Let $\calH(\Omega)=C^\omega(\Omega;\C)$ be the algebra (for the pointwise product) of complex -valued functions which are holomorphic on $\Omega$. Endowed with the topology of compact convergence
\footnote{
This topology is defined by the seminorms\\ 
\centerline{$p_K(f)=\sup_{s\in K}|f(s)|$
($K\subset \Omega$ is compact).}
}, it is a nuclear space\footnote{Space where commutatively convergent and absolutely convergent series are the same. This will allow the domain of the polylogarithm to be closed by shuffle products (i.e. the possiblity to compute legal polylogarithms through shuffle products).}.
\begin{definition}\label{domLi}
i) Let $S\in \ncs{\C}{X}$ be a series decomposed in its homogeneous (w.r.t. the length) components $S_n=\sum_{|w|=n}\scal{S}{w}\,w$ (so that $S=\sum_{n\geq 0}S_n$) is in the \textit{domain of $\Li$} iff the family 
$(\Li_{S_n})_{n\geq 0}$
%\begin{equation}
%\Big(\Li_{S_n}\Big)_{n\geq 0}
%\end{equation}  
is summable in $\calH(\Omega)$ in other words, due to the fact that the space is complete (see \cite{Sch}), if one has 
\begin{equation}\label{Cauchy}
(\forall W\in \mathcal{B}_{\calH(\Omega)})(\exists N)(\forall n\geq N)(\forall k)(\sum_{n\leq j\leq n+k}\Li_{S_j}\in W)\ .
\end{equation}
where $\mathcal{B}_{\calH(\Omega)}$ is the set of neighbourhoods of $0$ in $\calH(\Omega)$.

ii) The set of these series will be noted $Dom(\Li)$ and, for $S\in Dom(\Li)$, the sum 
$\sum_{n\geq 0} \Li_{S_n}$ will be noted $\Li_S$. 
\end{definition}
Of course criterium \ref{Cauchy} is only a theoretical tool to establish properties of the Domain of $\Li$. In further calculations (i.e. in practice), we will not use it but the stability of the domain under certain operations.
\begin{example}[\cite{Minh1}]
For example, the classical polylogarithms: dilogarithm 
$\Li_2$, trilogarithm $\Li_3$,
etc... are defined and obtained through the coding \eqref{re-index} by 
\begin{eqnarray*}
\Li_k(z)=\sum_{n\ge1}\frac{z^n}{n^k}=\Li_{x_0^{k-1}x_1}(z)=\scal{\L(z)}{x_0^{k-1}x_1}
\end{eqnarray*}
(where $\L(z)$ is as in Eq. \ref{polylog2}) 
but, one can check that, for $t\ge0$ (real), the series 
$(tx_0)^*x_1$ belongs to $\Dom(\Li_\bullet)$ (see Def. \ref{domLi}.ii) iff $0\le t<1$. In fact, in this case, 
$$
\Li_{(tx_0)^*x_1}(z) = \sum\limits_{n \geq 1} \dfrac{z^n}{n -t}.
$$ 
This opens the door of Hurwitz polyzetas \cite{Minh6}.
\end{example}

The map $\Li_{\bullet}$ is now extended to a subdomain of $\ncs{\C}{X}$,
called $\Dom(\Li_{\bullet})$ (see also \cite{CHM1,GHM1}).
%\begin{tcbstrikeout}
% It is the set of series
%\begin{eqnarray}
%S=\sum\limits_{n\ge0}S_n,&\mbox{where}&S_n:=\sum\limits_{|w|=n}\scal{S}{w}
%\end{eqnarray}
%such that $\sum\limits_{n\ge0}\Li_{S_n}$ is commutatively (or unconditionally) convergent for the standard topology on $\calH(\Omega)$ \cite{Sch}.\tcp{duplicate with def. \ref{domLi} todo (GHED) clear it}
%\end{tcbstrikeout}

\begin{example}
For any $\alpha , \beta \in \C$, 
$(\alpha x_0)^*, (\beta x_1)^*$, and  
$(\alpha x_0 + \beta x_1)^*=
(\alpha x_0)^*\shuffle (\beta x_1)^*$. We have 
\begin{eqnarray*}
&& \Li_{\alpha x_0^*} (z) = z^{\alpha}\ ;
\ \Li_{\beta x_1^*}(z)= (1-z)^{-\beta}\ ;\ 
 \Li_{(\alpha x_0 + \beta x_1)^*}(z) = 
 z^{\alpha} (1-z)^{-\beta}
\end{eqnarray*}
where $z \in \Omega$.
\end{example}

\begin{proposition} i) The domain $Dom(\Li)$ is a shuffle subalgebra of $(\ncs{\C}{X},\shuffle,1_{X^*})$.\\
ii) The extended polylogarithm 
$$
\Li:\ Dom(\Li)\to \calH(\Omega)
$$
is a shuffle morphism,i.e. 
 $S,T\in Dom(\Li)$, we still have 
\begin{equation}\label{Li_shuffle_mor}
\Li_{S\shuffle T}=\Li_{S}\Li_{T}\mbox{ and }
\Li_{1_{X^*}}=1_{\calH(\Omega)}
\end{equation}  
\end{proposition}
\begin{proof} This proof has been done in \cite{GHM1}.
\end{proof}

The picture about $Dom(\Li)$ within the algebra 
$(\ncs{\C}{X},\shuffle,1_{X^*})$, the positioning of 
$\ncs{\C^{\mathrm{rat}}}{X}$ (rational series, see \cite{berstel,CHM1,GHM1}) and shuffle subalegbras as, for example,\\ 
$\calA=\ncp{\C}{X}\shuffle \ncs{\C^{\mathrm{rat}}}{x_0}\shuffle \ncs{\C^{\mathrm{rat}}}{x_1}$ read as follows:
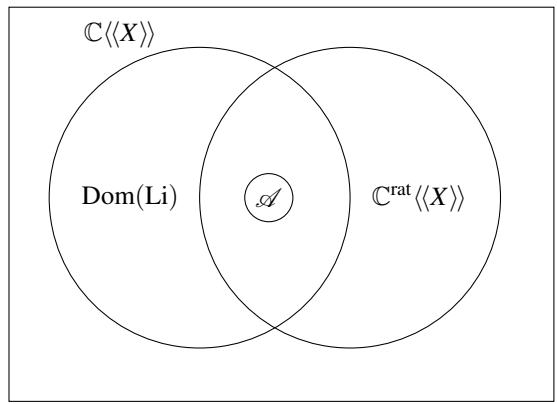
\vspace{4mm}
\begin{center}
\begin{tikzpicture}
\def\radius{2cm}
\def\mycolorbox#1{\textcolor{#1}{\rule{2ex}{2ex}}}
\colorlet{colori}{blue!60}
\colorlet{colorii}{red!60}
%
% some coordinates for the center of the circles
\coordinate (ceni);
\coordinate[xshift=.9\radius] (ceniii);
\coordinate[xshift=\radius] (cenii);

% the circles
\draw (ceni) circle (\radius);
\draw (cenii) circle (\radius);
\draw (ceniii) circle (0.3\radius);

% the rectangle
\draw  ([xshift=-15pt,yshift=15pt]current bounding box.north west) 
  rectangle ([xshift=20pt,yshift=-20pt]current bounding box.south east);

%the labels
\node[xshift=-.9\radius] at (ceni) {$\Dom(\Li)$};
\node[xshift=.9\radius] at (cenii) {$\ncs{\C^{\mathrm{rat}}}{X}$};
\node[xshift=.9\radius] at (ceni) {$\calA$};
\node[xshift=-30pt,yshift=\radius+5pt] at (ceni) {$\ncs{\C}{X}$};
\end{tikzpicture}
\end{center}
\section{From Polylogarithms to Harmonic sums}\label{removable}
Definition of $Dom(\Li)$ has many merits\footnote{As the fact that, due to special properties of $\calH(\Omega)$
(it is a nuclear space \cite{Sch}, see details in \cite{CHM1}), one can show that
$\Dom(\Li)$ is closed by shuffle products.} and can easily be adapted to arbitrary (open and connected)
domains. However this definition, based on a global condition of a fixed domain $\Omega$, does not provide
a sufficiently clear  interpretation of the stable symbolic computations around a point, in particular
at $z=0$. One needs to consider a sort of ``symbolic local germ'' worked out explicitely. Indeed,
as the harmonic sums (or MZV\footnote{Multiple Zeta Values.}) are the coefficients of the Taylor expansion at zero of the convergent
polylogarithms divided by $1-z$, we only need to know locally these functions. In order to gain more
indexing series and to describe the local situation at zero, we reshape and define a new domain of $\Li$
around zero to $\Dom^{\mathrm{loc}}(\Li_{\bullet})$.\\
The first step will be to characterize the polylogarithms having a removable singularity at zero
The following Proposition helps us characterize their indices.
\begin{proposition}\label{removable1}
Let $P\in \ncp{\C}{X}$ and  
$f(z)=\scal{\L}{P}=\sum_{w\in X^*}\scal{P}{w}\Li_w$.\\ 
1) The following conditions are equivalent
\begin{enumerate}
\item[i)] $f$ can be analytically extended around zero.
\item[ii)] $P\in \ncp{\C}{X}x_1\oplus \C.1_{X^*}$.
\end{enumerate}
2) In this case $\Omega$ itself\footnote{The domain, for 
$z$ of $\Li_P$.} can be extended to 
$\Omega_1=\C\setminus(]-\infty,-1]\cup[1,+\infty[)$.
\end{proposition}
\begin{proof}[Sketch] (ii) $\Longrightarrow$ (i) being straightforward, it remains to prove that  
(ii) $\Longrightarrow$ (i).\\ 
Let then $P\in \ncp{\C}{X}$ such that   
$f(z)=\scal{\L}{P}$ has a removable singularity at zero.
As a consequence of Radford's results \cite{Radford}, one can write down a basis of any free shuffle algebra in terms of Lyndon words. This implies that our polynomial reads 
\begin{equation}\label{RadDec}
P=\sum_{k\ge0}\alpha_k (P_k\shuffle x_0^{\shuffle\,k})\mbox{ with } 
\alpha_k\in \C,\ P_k\in \ncp{\C}{X}x_1\oplus \C.1_{X^*}
\end{equation}
the family $(P_k)_{k\ge0}$ being unique and finitely supported.\\
Using \eqref{RadDec} and \eqref{Li_shuffle_mor0}, we get 
$$
\Li_P(z)= \sum_{k\ge0}\alpha_k \Li_{P_k}(z)\log(z)^k
$$
the result now follows easily using asymptotic scale $x^n\log(x)^m$ along the axis $]0,+\infty[$ (and for 
$x\to 0_+$).   
\end{proof}

\bigskip
The second step will be provided by the following Proposition which says that, for appropriate series, the Taylor coefficients behave nicely.
\begin{proposition}\label{pre_dom_loc}
Let $S\in \ncs{\C}{X}x_1\oplus \C1_{X^*}$ such that
$S=\sum_{n\ge0}[S]_n$ where\\  
$[S]_n=\sum_{w\in X^*,\abs{w}=n}\scal{S}{w}w,$ ($[S]_n$ are the homogeneous components of $S$), we suppose that $0<R\le1$
and that $\sum\limits_{n\ge0}\Li_{[S]_n}$ is unconditionally convergent
(for the standard topology) within the open disk $\abs{z}<R$\footnote{With the definition given later \eqref{symb_loc} this amounts to say that\\ 
$
S\in \ncs{\C}{X}x_1\oplus \C1_{X^*}\cap Dom_R(\Li)\ .
$}.
Remarking that $\dfrac{1}{1-z}\sum\limits_{n\ge0}\Li_{[S]_n}(z)$
is unconditionally convergent in the same disk, we set 
\begin{eqnarray*}
\frac{1}{1-z}\sum_{n\ge0}\Li_{[S]_n}(z)=\sum_{N\ge0}
a_Nz^N\ .
\end{eqnarray*}
Then, for all $N\ge0$,
$
\sum\limits_{n\ge0}\H_{\pi_Y([S]_n)}(N)=a_N.
$
\end{proposition}
\begin{proof} 
Let us recall that, for any $w \in X^*$, the function $(1-z)^{-1}\Li_w(z)$ is analytic in the open disk $|z|<R$. Moreover, one has 
\begin{eqnarray*}
\frac{1}{1-z}\Li_w(z)=\sum\limits_{N\ge0}\H_{\pi_Y(w)}(N)z^N.
\end{eqnarray*}
Since $[S]_n=\sum\limits_{w\in X^*,\abs{w}=n}\scal{S}{w}w$ and $(1-z)^{-1}\sum\limits_{n\ge0}\Li_{[S]_n}$ 
absolutely converges (for the standard topology\footnote{For this topology, unconditional
and absolute convergence coincide \cite{Sch}.}.) within the open disk $D_{<R}$, one obtains, for all $|z|<R$
\begin{eqnarray*}
\frac{1}{1-z}\sum_{n\ge0}\Li_{[S]_n}(z)
&=&\frac{1}{1-z}\sum\limits_{n\ge0}\ \sum\limits_{w\in X^*,\abs{w}=n}\scal{S}{w}w\Li_w(z)=
\sum\limits_{n\ge0}\ \sum\limits_{w\in X^*,\abs{w}=n}\scal{S}{w}w\dfrac{\Li_w(z)}{1-z}\cr
&=&\sum\limits_{n\ge0}\ \sum\limits_{w\in X^*,\abs{w}=n}\scal{S}{w}w\sum\limits_{N\ge0}\H_{\pi_Y(w)}(N)z^N \cr
&\substack{(*)\\=}&\sum\limits_{N\ge0}\ \sum_{n\ge0}\ \sum\limits_{w\in X^*,\abs{w}=n}\scal{S}{w}w\H_{\pi_Y(w)}(N)z^N
=\sum\limits_{N\ge0}\H_{\pi_Y([S]_n)}(N)z^N.
\end{eqnarray*}
$(*)$ being possible because 
$\sum\limits_{w\in X^*,\abs{w}=n}$ is finite.\\
This implies that, for any $N\ge0$,
$
a_N=\sum\limits_{n\ge0}\H_{\pi_Y([S]_n)}(N).
$
\end{proof}

%\begin{lemma}\label{surj_comb}
%For a letter ``$a$'', one has 
%\begin{eqnarray}\label{shuffle_pow}
%\abs{(a^+)^{\shuffle m}}{a^n}=m!S_2(n,m)
%\end{eqnarray}
%($S_2(n,m)$ being the Stirling numbers of the second kind).
%The exponential generating series of R.H.S. in equation \eqref{shuffle_pow} (w.r.t. $n$) is given by
%\begin{eqnarray}
%\label{12fold1}
%\sum_{n\ge0} m! S_2(n,m)\frac{x^n}{n!}=(e^x-1)^m.
%\end{eqnarray}
%\end{lemma}
%\tcp{This lemma comes ``out of the blue''. Do not suppress it but rather explain (or ask me to do this :) ``Why'' you need it there.}
%
%\begin{proof}$(a^+)^{\shuffle m}$ is the specialization of 
%\begin{eqnarray*}
%L_m=a_1^+\shuffle a_2^+\shuffle\ldots\shuffle a_m^+
%\end{eqnarray*}
%to $a_j\to a$ (for all $j=1,2\ldots m$). The words of $L_m$ are in bijection with the surjections 
%$[1\ldots n]\to [1\ldots m]$, therefore the coefficient $\left\langle  (a^+)^{\shuffle m} |  a^n \right\rangle $ is exactly 
%the number of such surjections namely $m!S_2(n,m)$. A classical formula\footnote{See \cite{Ric},
%the twelvefold way, formula (1.94b)(pp. 74) for instance.} says that  
%\begin{eqnarray}
%\label{12fold2}
%\sum_{n\ge0} m! S_2(n,m)\frac{x^n}{n!}=(e^x-1)^m.
%\end{eqnarray}
%\end{proof}

To prepare the construction of the ``symbolic local germ'' around zero, let us set, in the same manner as in \cite{CHM1,GHM1},
\begin{eqnarray}\label{symb_loc}
\Dom_R(\Li)&:=&\{S\in\ncs{\C}{X}x_1\oplus \C1_{X^*}\vert\cr
&&\sum_{n\ge0}\Li_{[S]_n} \mbox{is unconditionally convergent in $\calH(D_{<R})$}\}
\end{eqnarray}
and prove the following:
\newpage
\begin{proposition}\label{lem2}
With the notations as above, we have: 
\begin{enumerate}
\item The map given by $R\mapsto\Dom_R(\Li)$ from 
$]0,1]$ to $2^{\ncs{\C}{X}}$ (the target is the set of subsets\footnote{For any set $E$, the set of its subsets is noted $2^E$.} of $\ncs{\C}{X}$ ordered by inclusion) is strictly decreasing
\item Each $\Dom_R(\Li)$ is a shuffle (unital) subalgebra of $\ncs{\C}{X}$.
\end{enumerate}
\end{proposition}

\begin{proof}
\begin{enumerate}
\item 
For $0<R_1<R_2\leq 1$ it is straightforward that
$
\Dom_{R_2}(\Li)\subset\Dom_{R_1}(\Li).
$
Let us prove that the inclusion is strict. Take $\abs{z}<1$ and let us, be it finite or infinite, evaluate the sum 
\begin{eqnarray*}
M(z)=\sum_{n\ge0}\abs{\Li_{[S]_n(t)}(z)}=\sum_{n\ge0}\scal{S(t)}{x_1^n}\abs{\Li_{x_1^n}(z)}
\end{eqnarray*}
then, by means of Lemma \ref{surj_comb}, with $x_1^+=x_1x_1^*=x_1^*-1$ and 
$S(t)=\sum\limits_{m\ge0}t^m(x_1^+)^{\shuffle m}$, we have
\begin{eqnarray*}
M(z)&=&\sum\limits_{n\ge0}\abs{S(t)}{x_1^n}\abs{\Li_{x_1^n}(z)}=\sum\limits_{n\ge0}\sum\limits_{m\ge0}\abs{t^m(x_1^+)^{\shuffle m}}{x_1^n}\abs{\Li_{x_1^n}(z)}\cr
&=&\sum\limits_{m\ge0}m!t^m\sum_{n\ge0}S_2(n,m)\frac{\abs{\Li_{x_1}(z)}^n}{n!}\le \sum\limits_{m\ge0}m!t^m\sum\limits_{n\ge0}S_2(n,m)\dfrac{\Li^n_{x_1}(\abs{z})}{n!},
\end{eqnarray*}
due to the fact that $\abs{\Li_{x_1}(z)}\le\Li_{x_1}(\abs{z})$ (Taylor series with positive coefficients).
Finally, in view of equation (\ref{12fold2}), we get, on the one hand, for $\abs{z}<(t+1)^{-1}$, 
\begin{eqnarray*}
M(z)\le\sum\limits_{m\ge0}t^m(e^{\Li_{x_1}(\abs{z})}-1)^m 
=\sum\limits_{m\ge0}t^m(\dfrac{\abs{z}}{1-\abs{z}})^m=\dfrac{1-\abs{z}}{1-(t+1)\abs{z}}.
\end{eqnarray*}
This proves that, for all $r\in ]0,\dfrac{1}{t+1}[$, 
$
\sum\limits_{n\ge0}\absv{\Li_{[S]_n(t)}(z)}_r<+\infty.
$

On the other hand, if $(t+1)^{-1}\le\abs{z}<1$, one has $M(|z|)=+\infty$,
and the preceding calculation shows that, with $t$ choosen such that
\begin{eqnarray*}
0\le\dfrac{1}{R_2}-1<t<\dfrac{1}{R_1}-1,
\end{eqnarray*}
we have $S(t)\in\Dom_{R_1}(\Li)$ but $S(t)\notin\Dom_{R_2}(\Li)$ whence, for $0<R_1<R_2\le1$, 
$\Dom_{R_2}(\Li) \subsetneq\Dom_{R_1}(\Li)$.

\item  One has (proofs as in \cite{GHM1})
\begin{enumerate}
\item $1_{X^*}\in\Dom_R(\Li)$ (because $1_{X^*}\in\ncp{\C}{X}$) and $\Li_{1_{X^*}}=1_{\calH(\Omega)}$.
\item Taking $S,T\in\Dom_R(\Li)$ we have, by absolute convergence, $S\shuffle T\in\Dom_R(\Li)$.
It is easily seen that $S\shuffle T\in\ncs{\C}{X}x_1\oplus \C1_{X^*} $ and, moreover, that 
$\Li_S\Li_T=\Li_{S\shuffle T}$\footnote{Proof by absolute convergence as in \cite{GHM1}.}. 
\end{enumerate}
\end{enumerate}
\end{proof} 

The combinatorial Lemma needed in the Theorem  \ref{PropDomR} is as follows 

\begin{lemma}\label{surj_comb}
For a letter ``$a$'', one has 
\begin{eqnarray}\label{shuffle_pow}
\abs{(a^+)^{\shuffle m}}{a^n}=m!S_2(n,m)
\end{eqnarray}
($S_2(n,m)$ being the Stirling numbers of the second kind).
The exponential generating series of R.H.S. in equation \eqref{shuffle_pow} (w.r.t. $n$) is given by
\begin{eqnarray}
\label{12fold1}
\sum\limits_{n\ge0} m! S_2(n,m)\dfrac{x^n}{n!}=(e^x-1)^m.
\end{eqnarray}
\end{lemma}
\begin{proof} The expression $(a^+)^{\shuffle m}$ is the specialization of 
\begin{eqnarray*}
L_m=a_1^+\shuffle a_2^+\shuffle\ldots\shuffle a_m^+
\end{eqnarray*}
to $a_j\to a$ (for all $j=1,2\ldots m$). The words of $L_m$ are in bijection with the surjections 
$[1\ldots n]\to [1\ldots m]$, therefore the coefficient $\left\langle  (a^+)^{\shuffle m} |  a^n \right\rangle $ is exactly 
the number of such surjections namely $m!S_2(n,m)$. A classical formula\footnote{See \cite{RPS},
the twelvefold way, formula (1.94b)(pp. 74) for instance.} says that  
\begin{eqnarray}
\label{12fold2}
\sum\limits_{n\ge0} m! S_2(n,m)\dfrac{x^n}{n!}=(e^x-1)^m.
\end{eqnarray}
\end{proof}

In Theorem \ref{PropDomR} below, we study, for series taken in $\ncs{\C}{X}x_1\oplus\C.1_{X^*}$, the correspondence 
$\Li_\bullet$ to some $\calH(D_{<R})$, first (point 1) establishes its surjectivity (in a certain sense) and then
(points 2 and 3) examine the relation between summability of the functions and that of their Taylor coefficients.
For that, let us begin with a very general Lemma on sequences of Taylor series which adapts, for our needs,
the notion of \textit{normal families} as in \cite{Montel}.
\begin{lemma}\label{TaylSeq}
Let $\tau=(a_{n,N})_{n,N\ge0}$ be a double sequence of complex numbers. Setting
\begin{eqnarray*}
I(\tau):=\{r\in]0,+\infty[\vert\sum_{n,N\ge0}|a_{n,N}r^N|<+\infty\},
\end{eqnarray*}
one has
\begin{enumerate} 
\item $I(\tau)$ is an interval of $]0,+\infty[$, it is not empty iff there exists $z_0\in\C\setminus\{0\}$ such that 
\begin{eqnarray}\label{NonZeroRad}
\sum_{n,N\ge0}|a_{n,N}z_0^N|<+\infty
\end{eqnarray}
In this case, we set $R(\tau):=\sup(I(\tau))$, one has 
\begin{enumerate}
\item For all $N$, the series $\sum\limits_{n\ge0}a_{n,N}$ converges absolutely (in $\C$). 
Let us note $a_N$ - with one subscript - its limit
\item For all $n$, the convergence radius of the Taylor series 
$
T_n(z)=\sum_{N\ge0}a_{n,N}z^N
$
is at least $R(\tau)$ and $\sum\limits_{n\in\N}T_n$ is summable for the standard
topology of $\calH(D_{<R(\tau)})$ with sum $T(z)=\sum\limits_{n,N\ge0}a_{N}z^N$.
\end{enumerate}

\item Conversely, we suppose that it exists $R>0$ such that 
\begin{enumerate}
\item Each Taylor series $T_n(z)=\sum\limits_{N\ge0}a_{n,N} z^N$ converges in $\calH(D_{<R})$.
\item The series $\sum\limits_{n\in\N}T_n$ converges unconditionnally in $\calH(D_{<R})$.
\end{enumerate}
Then $I(\tau)\not=\emptyset$ and $R(\tau)\ge R$.
\end{enumerate}
\end{lemma}

\begin{proof}
\begin{enumerate}
\item The fact that $I(\tau)\subset]0,+\infty[$ is straightforward from the Definition.
If it exists $z_0\in \C$ such that
$
\sum_{n,N\ge0}\abs{a_{n,N}z_0^N}<+\infty
$
then, for all $r\in]0,|z_0|[$, we have 
\begin{eqnarray*}
\sum_{n,N\ge0}\abs{a_{n,N}r^N}
=\sum_{n,N\ge0}\abs{a_{n,N}z_0^N}\biggl(\frac{r}{\abs{z_0}}\biggr)^N
\le\sum_{n,N\ge0}|a_{n,N}z_0^N|<+\infty
\end{eqnarray*}
in particular $I(\tau)\not=\emptyset$ and it is an interval of $]0,+\infty[$ with lower bound zero.     
\begin{enumerate}
\item Take $r\in I(\tau)$ (hence $r\not=0$) and $N\in \N$, then we get the expected result as
\begin{eqnarray*}
r^N\sum_{n\ge0}\abs{a_{n,N}}=\sum_{n\ge0}\abs{a_{n,N}r^N}\le\sum_{n,N\ge0}\abs{a_{n,N}r^N}<+\infty.
\end{eqnarray*}
\item Again, take any $r\in I(\tau)$ and $n\in\N$, then 
$
\sum\limits_{N\ge0}\abs{a_{n,N}r^N}<+\infty
$
which proves that $R(T_n)\ge r$, hence the result\footnote{For a Taylor series $T$,
we note $R(T)$ the radius of convergence of $T$.}. We also have 
\begin{eqnarray*}
\abs{\sum\limits_{N\ge0}a_Nr^N}\le\sum\limits_{N\ge0}r^N\abs{\sum\limits_{n\ge0}a_{n,N}}\le\sum\limits_{n,N\ge0}\abs{a_{n,N}r^N}<+\infty
\end{eqnarray*}
and this proves that $R(T)\ge r$, hence $R(T)\ge R(\tau)$.
\end{enumerate}
\item Let $0<r<r_1<R$ and consider the path $\gamma(t)=r_1e^{2i\pi t}$, we have
\begin{eqnarray*}
\abs{a_{n,N}}=\abs{\dfrac{1}{2i\pi}\int_\gamma\dfrac{T_n(z)}{z^{N+1}}dz}\le 
\dfrac{2\pi}{2\pi}\dfrac{r_1\absv{T_n}_{K}}{r_1^{N+1}}\le\dfrac{\absv{T_n}_K}{r_1^{N}} 
\end{eqnarray*} 
with $K=\gamma([0,2\pi])$, hence
\begin{eqnarray*}
\sum\limits_{n,N\ge0}\abs{a_{n,N}r^N}\le\sum\limits_{n,N\ge0}\abs{T_n}_K(\dfrac{r}{r_1})^N\le\dfrac{r_1}{r_1-r}\sum\limits_{n\ge0}\absv{T_n}_K<+\infty.
\end{eqnarray*}
\end{enumerate}
\end{proof}

\begin{remark}
\begin{enumerate}
\item[(i)] First point says that every function analytic at zero can be represented around zero as $\Li_S(z)$ for 
some $S\in\ncs{\C}{x_1}$.
\item[(ii)] In point 2, the arithmetic functions 
$\H_{\pi_{Y}(S)}\in \Q^\N$, for $S\in Dom(\Li)$ are quickly defined (and in a way extending the old definition) and we draw a very important bound saying  that, in this condition, for some $r>0$ the array 
$\left(\H_{\pi_Y([S]_n)}(N)r^N\right)_{n,N}$ converges (then, in particular, horizontally and vertically).
\item[(iii)] Point 3 establishes the converse.   
\end{enumerate}
\end{remark}

%\newpage  
\begin{theorem}\label{PropDomR}
\begin{enumerate} 
\item Let $T(z)=\sum\limits_{N\ge0}a_Nz^N$ be a Taylor series converging on some non-empty disk centered at zero \textit{i.e.} such that 
$
\limsup_{N\to+\infty}\abs{a_N}^{\frac{1}{n}}=B<+\infty,
$
then the series 
\begin{eqnarray}\label{preim_series1}
S=\sum_{N\ge0}a_N(-(-x_1)^+)^{\shuffle N}
\end{eqnarray} 
is summable  in $\ncs{\C}{X}$ (with sum in $\ncs{\C}{x_1}$),
$S\in\Dom_R(\Li)$ with $R=(B+1)^{-1}$ and $\Li_S=T$.
\item Let $S\in\Dom_{R}(\Li)$ and $S=\sum\limits_{n\ge0}[S]_n$ (homogeneous decomposition), we 
define\\ $N\longmapsto\H_{\pi_Y(S)}(N)$ by\footnote{This Definition is compatible with the old one 
when $S$ is a polynomial.}
\begin{equation}\label{absconv1}
\dfrac{\Li_S(z)}{1-z}=\sum\limits_{N\ge0}\H_{\pi_{Y}(S)}(N)z^N.
\end{equation} 
Then,
\begin{eqnarray}\label{abs_sum0}
\forall r\in]0,R[,&&\sum\limits_{n,N\ge0}\abs{\H_{\pi_Y([S]_n)}(N)r^N}<+\infty.
\end{eqnarray} 
In particular, for all $N \in \N$, the series (of complex numbers), 
$\sum\limits_{n\ge0}\H_{\pi_Y([S]_n)}(N)$ converges absolutely to $\H_{\pi_Y(S)}(N)$. 
\item\label{QinDomH}  Conversely, let $Q \in \ncs{\C}{Y}$ with $Q=\sum\limits_{n\ge0}Q_n$
(decomposition by weights), we suppose that it exists $r\in ]0,1]$ such that 
\begin{eqnarray}\label{abs_sum2}
\sum\limits_{n,N\ge0}\abs{\H_{Q_n}(N)r^N}<+\infty,
\end{eqnarray}
in particular, for all $N\in \N$, $\sum\limits_{n\ge0}\H_{Q_n}(N)=\ell(N)\in\C$
unconditionally. Under such circumstances, $\pi_X(Q)\in\Dom_r(\Li)$ and, for all $z\in\C,\abs{z}\le r$,
\begin{eqnarray}\label{L2H_corresp2}
\frac{\Li_S(z)}{1-z}=\sum_{N\ge0}\ell(N)z^N,
\end{eqnarray}
\end{enumerate}
\end{theorem}
\begin{proof}
\begin{enumerate}
\item The fact that the series \eqref{preim_series1} is summable comes from the fact that\\ 
$\omega(a_N(-(-x_1)^+)^{\shuffle N})\ge N$. Now from the Lemma \ref{surj_comb}, we get 
$$
(S)_n=\sum_{N\ge0}(a_N(-(-x_1)^+)^{\shuffle N})_n=(-1)^{N+n}a_NN!S_2(n,N)x_1^n.
$$
Then, with $r=\sup_{z\in K}\abs{z}$ (we have indeed $r=||Id||_K$) and
taking into account that $\absv{Li_{x_1}}_K\le\log({1}/(1-r))$, we have
\begin{eqnarray*}
\sum_{n\ge0}\absv{\Li_{(S)_n}}_K
&\le&\sum_{n\ge0}\sum_{N\ge0}\abs{a_N}N!S_2(n,N)\absv{\Li_{x_1^n}}_K
\le \sum_{n\ge0}\sum_{N\ge0}\abs{a_N}N!S_2(n,N)\frac{\absv{\Li_{x_1}}_K^n}{n!} \cr 
&\le&\sum_{N\ge0}\abs{a_N}\sum_{n\ge0}N!S_2(n,N)\frac{\abs{\Li_{x_1}}_K^n}{n!} 
\le\sum_{N\ge0}\abs{a_N}(e^{\log(\frac{1}{1-r})}-1)^N \cr
&=&\sum_{N\ge0}\abs{a_N}\biggl(\frac{r}{1-r}\biggr)^N.
\end{eqnarray*}
Now if we suppose that $r\le(B+1)^{-1}$, we have $r(1-r)^{-1}\le \dfrac{1}{B}$
and this shows that the last sum is finite.
\item This point and next point are consequences of Lemma \ref{TaylSeq}. Now, considering the homogeneous decomposition
$
S=\sum\limits_{n\ge0}[S]_n\in\Dom_R(\Li).
$
We first establish inequation \eqref{abs_sum0}. Let $0<r<r_1<R$ and consider the path 
$\gamma(t)=r_1e^{2i\pi t}$, we have
\begin{eqnarray*}
\abs{\H_{\pi_Y([S]_n)}(N)}=\abs{\frac{1}{2i\pi}\int_\gamma\frac{\Li_{[S]_n}(z)}{(1-z)z^{N+1}}dz}\le 
\frac{2\pi}{2\pi}\frac{\absv{\Li_{[S]_n}}_K}{(1-r_1)r_1^{N+1}},
\end{eqnarray*}
$K=\gamma([0,1])$ being the circle of center $0$ and radius $r_1$.
Taking into account that, for $K\subset_{compact}\D_{<R}$, we have a decomposition
$
\sum_{n\in\N}\abs{\Li_{[S]_n}}_K=M<+\infty,
$
we get 
\begin{eqnarray*}
\sum_{n,N\ge0}\abs{\H_{\pi_Y([S]_n)}(N)r^N}
&=&\sum_{n,N\ge0}\abs{\H_{\pi_Y([S]_n)}(N)r_1^N}(\frac{r}{r_1})^N 
=\sum_{N\ge0}(\frac{r}{r_1})^N\sum_{n\ge0}\abs{\H_{\pi_Y([S]_n)}(N)r_1^N}\cr
&\le& \sum_{N\ge0}(\frac{r}{r_1})^N\frac{M}{(1-r_1)r_1}\le \frac{M}{(1-r_1)(r_1-r)}<+\infty.
\end{eqnarray*}  
The series 
$\sum\limits_{n\ge0}\Li_{[S]_n}(z)$ converges to $\Li_{S}(z)$ in $\calH(D_{<R})$
($D_{<R}$ is the open disk defined by $|z|<R$). For any $N\ge0$, by Cauchy's formula, one has, 
\begin{eqnarray*}
\H_{\pi_Y(S)}(N)
&=&\frac{1}{2i\pi}\int_\gamma\frac{\Li_{S}(z)}{(1-z)z^{N+1}}dz=\frac{1}{2i\pi}\int_\gamma\frac{\sum_{n\ge0}\Li_{[S]_n}(z)}{(1-z)z^{N+1}}dz \cr &=&\frac{1}{2i\pi}\sum_{n\ge0}\int_\gamma\frac{\Li_{[S]_n}(z)}{(1-z)z^{N+1}}dz=\sum_{n\ge0}\H_{\pi_Y([S]_n)}(N)
\end{eqnarray*}
the exchange of sum and integral being due to the compact convergence.
The absolute convergence comes from the fact that the convergence of
$\sum\limits_{n\ge}\Li_{[S]_n}(z)$ is unconditional \cite{Sch}.
\item Fixing $N\in\N$, from inequation \eqref{abs_sum2}, we get
$\sum\limits_{n\ge0}\abs{\H_{Q_n}(N)}<+\infty$ which proves the absolute convergence.
Remark now that $(\pi_X(Q))_n=\pi_X(Q_n)$ and $\pi_Y(\pi_X(Q_n))=Q_n$,
one has, for all $\abs{z}\le r$, 
$
\abs{\Li_{\pi_X(Q_n)}(z)}=\abs{\sum\limits_{N\in\N}\H_{Q_n}(N)z^N}\le\abs{\sum\limits_{N\in \N}\H_{Q_n}(N)r^N},
$
in other words
$
\absv{\Li_{\pi_X(Q_n)}}_{D\leq r}\le\abs{\sum\limits_{N\in \N}\H_{Q_n}(N)r^N}
$
and 
\begin{eqnarray*}
\sum_{n\in \N}\absv{\Li_{\pi_X(Q_n)}}_{D\le r}\le\abs{\sum_{n,N\in\N}\H_{Q_n}(N)r^N}<+\infty
\end{eqnarray*}
which shows that $\pi_X(Q)\in\Dom_r(\Li)$. The equation (\ref{L2H_corresp2}) is a consequence of\\ point 2, taking $S=\pi_X(Q)$.
\end{enumerate}
\end{proof}
Now, we have have a better understanding of what can (and will) be the domain, $\Dom(\H_\bullet)$, of harmonic sums.  
\begin{definition}\label{def2}
We set 
$
\Dom^{\mathrm{loc}}(\Li)=\bigcup\limits_{0<R\le1}\Dom_R(\Li) ; 
\Dom(\H_\bullet)=\pi_Y(\Dom^{\mathrm{loc}}(\Li))
$
and, for $S\in\Dom^{\mathrm{loc}}(\Li)$, $\Li_S(z)=\sum\limits_{n\ge0}\Li_{[S]_n}(z)$ and $\dfrac{\Li_S(z)}{1-z}=\sum\limits_{N\ge0}\H_{\pi_Y(S)}(N)z^N$.
\end{definition}
\section{Applications}
\quad  We remark that formula \eqref{polylogarithm}, {\it i.e.}, $
\Li_{s_1,\ldots,s_r}(z):=\sum\limits_{n_1>\ldots>n_r>0}\dfrac{z^{n_1}}{n_1^{s_1}\ldots n_r^{s_r}}$,
still makes sense for $|z|<1$ and $(s_1,\ldots,s_r)\in \C^r$ so that we will freely use the list indexing to get index lists with $s_i\in \Z$ for any $i =1, \ldots , r$ and $r \in \N^{+}$.

Recall that for any $s_1, \ldots , s_r  \in \N$, we can present $\Li_{-s_1, \ldots , -s_r} (z)$ as a polynomial of $\dfrac{1}{1-z}$ with integer coefficients. Then, using \eqref{Li_shuffle_mor} and 
$(kx_1)^*=[(x_1)^*]^{\shuffle\,k}$, we get 
$
\dfrac{1}{(1-z)^k} =\Li_{(kx_1)^*}(z), \  \forall k \in \N^+,
$ 
we obtain a polynomial 
$P \in \Dom(\Li) \cap \C[x_1^*]=\C[x_1^*]$ such that $\Li_{-s_1, \ldots , -s_r} = \Li_P$ (see \cite{GHM1}).  Using Theorem \ref{PropDomR}, we have $ \dfrac{\Li_P(z)}{1-z} = \sum\limits_{N\geq 0} \H_{\pi_Y(P)} (N) z^N $. This means that we can provide a class of elements of 
$\Dom(\H_\bullet)$ (as in Definition \ref{def2}) relative to the set of indices of harmonic sums at negative integer multiindices. Here are some examples.
\begin{example}\label{ex3}
For any $|z| < 1$, we have
\begin{eqnarray*}
&&\Li_{x_1^*}(z) = \dfrac{1}{1-z}\ ;\ 
\Li_{x_1^* -1_{X^*}} (z) = \dfrac{z}{1-z} = \Li_{0} (z)\ 
;\ \Li_{(2x_1)^* -x_1^* } (z) = \dfrac{z}{(1-z)^2}  = \Li_{-1}(z);\\
&&\Li_{(2x_1)^* -2 x_1^* +1_{X^*}} (z) = \dfrac{z^2}{(1-z)^2} = \Li_{0,0}(z);\\
&&\Li_{12(5x_1)^* - 33(4x_1)^* +31(3x_1)^* - 11(2x_1)^* +x_1^*} (z) = \dfrac{z^4 +7z^3 +4z^2}{(1-z)^5} = \Li_{-2,-1}(z);\\
&&\Li_{40(6x_1)^* - 132(5x_1)^* + 161(4x_1)^* -87(3x_1)^* +19(2x_1)^* -x_1^*} (z) = \dfrac{z^5 + 14z^4 + 21z^3 + 4z^2}{(1-z)^6} = \Li_{-2,-2}(z);
\end{eqnarray*}
\begin{eqnarray*}
&&\Li_{1260(8x_1)^* - 5400(7x_1)^* + 9270(6x_1)^* -8070 (5x_1)^* + 3699 (4x_1)^* - 829 (3x_1)^* + 71 (2x_1)^* - x_1^*  } (z)\\
&& = \dfrac{ z^7 + 64z^6 + 424z^5 + 584z^4 + 179z^3 + 8z^2}{(1-z)^8} = \Li_{-3,-3}(z);\\
&&\Li_{10(6x_1)^* - 38(5x_1)^* +55(4x_1)^*- 37(3x_1)^* +11(2x_1)^*-x_1^*}(z)=  \dfrac{z^5 + 6z^4 + 3z^3}{(1-z)^6} = \Li_{-1,0,-2}(z);\\
&& \Li_{280(8x_1)^* -1312(7x_1)^* +2497(6x_1)^* - 2457(5x_1)^* +1310(4x_1)^*-358(3x_1)^*+41(2x_1)^*-x_1^*}(z)\\
&& = \dfrac{z^7 + 34z^6 + 133z^5 + 100z^4 + 12z^3}{(1-z)^8} = \Li_{-1,-2,-2}(z).
\end{eqnarray*}
Thus, for any $N \in \N$, for readability, below $1$ stands for $1_{X^*}$
\begin{eqnarray*}
&&\H_{\pi_Y(x^*_1)} (N) = N +1 ,\ \H_{\pi_{Y} (x_1^* - 1)}(N) = N = \sum\limits_{n=1}^N n^0 ,\ \H_{\pi_Y ((2x_1)^* -x_1^*)}  (N) = \dfrac{1}{2}N^2 + \dfrac{1}{2}N = \sum\limits_{n=1}^N n^1; \\
&&\H_{\pi_Y ((2x_1)^* -2 x_1^* +1)}  (N) = \dfrac{1}{2}N^2 - \dfrac{1}{2}N = \sum\limits_{n_1=1}^N n_1^0 \sum_{n_2=1}^{n_1-1} n_2^0 ;\\
&&\H_{\pi_Y( 12(5x_1)^* - 33(4x_1)^* +31(3x_1)^* - 11(2x_1)^* +x_1^*)} (N) = \dfrac{1}{10}N^5+\dfrac{ 1}{8}N^4  - \dfrac{ 1}{12}N^3 - \dfrac{ 1}{60}N - \dfrac{1}{8}N^2=  \sum\limits_{n_1=1}^N n_1^2 \sum_{n_2=1}^{n_1-1} n_2^1; \\
&&\H_{\pi_Y(40(6x_1)^* - 132(5x_1)^* + 161(4x_1)^* -87(3x_1)^* +19(2x_1)^* -x_1^*)} (N) = \dfrac{1}{15}N^5 + \dfrac{1}{18}N^6 - \dfrac{ 5}{72}N^4 + \dfrac{1}{72}N^2 \\
&&+ \dfrac{ 1}{60}N - \dfrac{1}{12}N^3 =   \sum\limits_{n_1=1}^N n_1^2 \sum_{n_2=1}^{n_1-1} n_2^2; \\
&&\H_{\pi_Y(1260(8x_1)^* - 5400(7x_1)^* + 9270(6x_1)^* -8070 (5x_1)^* + 3699 (4x_1)^* - 829 (3x_1)^* + 71 (2x_1)^* - x_1^*)  } (N) = \sum\limits_{n_1=1}^N n_1^3 \sum_{n_2=1}^{n_1-1} n_2^3;\\
&&\H_{\pi_Y(10(6x_1)^* - 38(5x_1)^* +55(4x_1)^*- 37(3x_1)^* +11(2x_1)^*-x_1^*)}(N) = -\dfrac{ 1}{40}N^5 +\dfrac{ 1}{72}N^6 - \dfrac{ 1}{36}N^4 + \dfrac{1}{72}N^2\\
&& + \dfrac{1}{24}N^3 - \dfrac{1}{60}N = \sum\limits_{n_1=1}^N n_1^1 \sum_{n_2=1}^{n_1-1} \sum_{n_3=1}^{n_2-1} n_3^2; \\
&&\H_{\pi_Y(280(8x_1)^* -1312(7x_1)^* +2497(6x_1)^* - 2457(5x_1)^* +1310(4x_1)^*-358(3x_1)^*+41(2x_1)^*-x_1^*)}(N) = - \dfrac{ 13}{1260}N^7 \\
&&+ \dfrac{ 1}{144}N^8 - \dfrac{ 7}{240}N^6 + \dfrac{ 1}{24}N^4 - \dfrac{ 7}{360}N^2 + \dfrac{ 23}{720}N^5 + \dfrac{ 1}{210}N - \dfrac{ 19}{720}N^3  = \sum\limits_{n_1=1}^N n_1^1 \sum_{n_2=1}^{n_1-1} n_2^2 \sum_{n_3=1}^{n_2-1} n_3^2 .
\end{eqnarray*}
\end{example}
Observe that, from this Definition, Theorem \ref{theomain} will show us that $\Dom(\H_\bullet)$ is a stuffle subalgebra of $\ncs{\C}{Y}$.  Let us however remark that some series are not in this domain as shown below
\begin{enumerate}
\item[(i)] The series $T=\sum\limits_{n=1}^{\infty}(-1)^{n-1} \dfrac{y_n}{n} \in \ncs{\C}{Y}$ is not in 
$\Dom(\H_\bullet)$ because we see that its decomposition by weights ($T=\sum\limits_{n=1}^{\infty} T_n$ 
as in \eqref{abs_sum2}) provides $T_n= \dfrac{(-1)^{n-1}}{n} y_n$ \ul{for} $n\ge1$ and $T_0=0$. Direct calculation, gives, for $n\ge1$  
$
\H_{y_n}(N) = \sum\limits_{k=1}^N \dfrac{1}{k^n} 
$,
so that we have
$\H_{y_n} (N) \geq 1, \forall n \in \N^+; N \in \N^+$, because $\H_{y_n}(0)=0$,  for all $0<r<1$, one has 
\begin{equation}
\sum\limits_{n,N}\abs{\H_{T_n}(N)r^N} = \sum\limits_{N \geq 0} 
\sum\limits_{n \geq 1}\abs{\dfrac{1}{n} \H_{y_n}(N)r^N} \ge 
\left( \sum\limits_{n\ge0} \frac{1}{n} \right) \frac{r}{1-r}  =  +\infty . 
\end{equation}

%\tcp{Difficult to read: What is $T_n$ ? is it of the same type as $T$ ? }\\
However one can get unconditional convergence using a sommation by pairs (odd + even).
\item[(ii)] For all $s\in]1,+\infty[$, the series $T(s)=\sum\limits_{n=1}^{\infty}(-1)^{n-1}y_nn^{-s}\in\ncs{\C}{Y}$ is in $\Dom(\H_\bullet)$.
\end{enumerate}

We can now state the 
\begin{theorem}\label{theomain}
Let $S,T\in\Dom^{\mathrm{loc}}(\Li)$, then
$
S\shuffle T\in\Dom^{\mathrm{loc}}(\Li),\pi_X(\pi_Y(S)\stuffle\pi_Y(T))\in\Dom^{\mathrm{loc}}(\Li)
$
and for all $N\ge0$,
\begin{eqnarray}
\Li_{S\shuffle T}&=&\Li_{S}\Li_{T};\quad\Li_{1_{X^*}}=1_{\calH(\Omega)},\label{eq1}\\
\H_{\pi_Y(S)\stuffle \pi_Y(T)}(N)&=& \H_{\pi_Y(S)}(N)\H_{\pi_Y(T)}(N).\label{eq2}\\
\dfrac{\Li_{S}(z)}{1-z}\odot\dfrac{\Li_{T}(z)}{1-z}&=&\dfrac{\Li_{\pi_X(\pi_Y(S)\stuffle\pi_Y(T))}(z)}{1-z}.\label{eq3}
\end{eqnarray}
\end{theorem}

\begin{proof}
For equation (\ref{eq1}), we get, from Lemma \ref{lem2} that $Dom^{loc}(\Li)$ is the union of an increasing
set of shuffle subalgebras of $\ncs{\C}{X}$. It is therefore a shuffle subalgebra of the latter. 

For equation \eqref{eq2}, suppose $S\in\Dom_{R_1}(\Li)$ (resp. $T\in\Dom_{R_2}(\Li)$).
By \cite{Had} and Theorem \ref{PropDomR}, one has 
$
\dfrac{\Li_{S}(z)}{1-z}\odot\dfrac{\Li_{T}(z)}{1-z}\in\Dom_{R_1R_2}(\Li)
$
where $\odot$ stands for the Hadamard product \cite{Had}. Hence, for $|z|<R_1R_2$, one has
\begin{eqnarray}
f(z)=\dfrac{\Li_{S}(z)}{1-z}\odot \dfrac{\Li_{T}(z)}{1-z}=\sum_{N\ge0}\H_{\pi_Y(S)}(N)\H_{\pi_Y(T)}(N)z^N
\end{eqnarray}
and, due to Theorem \ref{PropDomR} point \eqref{absconv1}, for all $N$,
$\sum\limits_{p\ge0}\H_{\pi_Y(S_p)}(N)=\H_{\pi_Y(S)}(N)$ and
$\sum\limits_{q\ge0}\H_{\pi_Y(T_q)}(N)=\H_{\pi_Y(T)}(N)$ 
(absolute convergence)
then, as the product of two absolutely convergent series is absolutely convergent
(w.r.t. the Cauchy product), one has, for all $N$,
\begin{eqnarray}
\H_{\pi_Y(S)}(N)\H_{\pi_Y(T)}(N)
&=&\biggl(\sum_{p\ge0}\H_{\pi_Y(S_p)}(N)\biggr)\biggl(\sum_{q\ge0}\H_{\pi_Y(T_q)}(N)\biggr) \cr
&=&\sum_{p,q\ge0}\H_{\pi_Y(S_p)}(N)\H_{\pi_Y(T_q)}(N)=\sum_{n\ge0}\sum_{p+q=n}\H_{\pi_Y(S_p)\stuffle\pi_Y(T_q)}(N)\cr
&=&\sum_{n\ge0}\H_{(\pi_Y(S)\stuffle\pi_Y(T))_n}(N). 
\end{eqnarray}

Remains to prove that condition of Theorem \ref{PropDomR}, \textit{i.e.} inequation \eqref{abs_sum2}
is fulfilled. To this end, we use the well-known fact that if $\sum_{m\ge 0}c_m\,z^m$ has radius of
convergence $R>0$, then $\sum\limits_{m\ge 0}\abs{c_m}z^m$ has the same radius of convergence
(use $1/R=\limsup_{m\ge 1}\abs{c_m}^{-m}$), then from the fact that $S\in\Dom_{R_1}(\Li)$ 
(resp. $T\in Dom_{R_2}(\Li)$), we have \eqref{abs_sum0} for each of them and,
using the Hadamard product of these expressions, we get\\
\centerline{$\forall r\in]0,R_1.R_2[,\sum_{p,q,N\ge0}|\H_{\pi_Y(S_p)}(N)\H_{\pi_Y(T_q)}(N)\,r^N|<+\infty,$}
and this assures, for $|z|<R_1R_2$, the convergence of 
\begin{eqnarray}
f(z)=\sum_{n,N\ge0}\H_{(\pi_Y(S)\stuffle\pi_Y(T))_n}(N)z^N
\end{eqnarray}
applying Theorem \ref{PropDomR} point \eqref{QinDomH} to $Q=\pi_Y(S)\stuffle\pi_Y(T)$
(with any $r<R_1R_2$), we get $\pi_X(Q)=\pi_X(\pi_Y(S)\stuffle\pi_Y(T))\in\Dom^{\mathrm{loc}}(\Li)$ and 
\begin{eqnarray*}
f(z)=\sum_{N\ge0}\Big(\sum_{n\ge0}\H_{(\pi_Y(S)\stuffle\pi_Y(T))_n}(N)\Big)z^N
=\dfrac{\Li_{\pi_X(\pi_Y(S)\stuffle\pi_Y(T))}(z)}{1-z}.
\end{eqnarray*}
hence  we obtain  \eqref{eq2}.
\end{proof}

Recall that, as in Example \ref{ex3},  for any $s_1, \ldots , s_r \in \N$, we can find an elements $P \in \Dom(\Li)$ such that $\dfrac{\Li_P(z)}{1 - z} = \sum\limits_{N \geq 0} \H_{\pi_Y(P) (N)} z^N.$ Theorem  \ref{theomain} proves that $\H$ is a stuffle character on $\Dom (\H)$. Then for any mixed multiindices $\mathrm{s}$, we can find the elements $P \in \Dom(\Li)$ satisfying $\dfrac{\Li_{\mathrm{s}}(z)}{1-z} = \sum\limits_{N\geq 0} \H_{\pi_Y(P)}(N) z^N$. 
\begin{example}
\begin{eqnarray*}
&&\H_{\pi_Y( \frac{1}{2} (2x_1)^* - x_1^* + \frac{1}{2}   )} (N) = \dfrac{1}{4}N^2 - \dfrac{1}{4}N = \sum\limits_{n_1=1}^N \dfrac{1}{n_1} \sum_{n_2=1}^{n_1-1} n_2^1, \\ 
\mbox{hence} && \H_{\pi_{Y}(x_1) \stuffle \pi_Y ((2x_1)^* -x_1^*)  -\frac{1}{2} \pi_Y(  (2x_1)^*  - 1  )} (N) =  \sum\limits_{n_1=1}^N n_1 \sum_{n_2=1}^{n_1-1} \dfrac{1}{ n_2}, \\
&& \H_{\pi_Y(\frac{2}{3} (3x_1)^*  -\frac{3}{2} (2x_1)^* +x_1^* -\frac{1}{6} )} (N) = \dfrac{1}{9}N^3 -\dfrac{1}{12}N^2 -\dfrac{1}{36}N =  \sum\limits_{n_1=1}^N \dfrac{1}{n_1} \sum_{n_2=1}^{n_1-1} n_2^2,\\
\mbox{then} &&  \H_{\pi_Y(x_1 ) \stuffle \pi_Y (2(3x_1)^* -3(2x_1 )^* +x_1^*)  - \pi_Y (\frac{2}{3}(3x_1)^* -\frac{1}{2} (2x_1)^* -\frac{1}{6}) } (N) = \sum\limits_{n_1=1}^N n_1^2 \sum_{n_2=1}^{n_1-1} \dfrac{1}{ n_2}, \\
&& \H_{\pi_{Y} (\frac{1}{3}(2x_1)^* - \frac{5}{6}x_1^* +\frac{1}{2} + \frac{1}{6} x_1)} (N) = \sum\limits_{n_1=1}^N \dfrac{1}{n_1^2} \sum_{n_2=1}^{n_1-1} n_2^2, \\
\mbox{which entails} && \H_{\pi_Y(x_0x_1) \stuffle \pi_Y ( 2(3x_1)^* -3(2x_1 )^* +x_1^*) -\pi_Y(\frac{1}{3}(2x_1)^* +\frac{1}{6}x_1^* -\frac{1}{2} + \frac{1}{6}x_1 )} = \sum\limits_{n_1=1}^N n_1^2 \sum_{n_2=1}^{n_1-1} \dfrac{1}{ n_2^2}, \\
&& \H_{\pi_Y(\frac{20}{3} (6x_1)^* -\frac{128}{5}(5x_1)^* +\frac{153}{4} (4x_1)^* -\frac{82}{3}(3x_1)^* +\frac{653}{72}(2x_1)^*  -\frac{373}{360} x_1^* - \frac{1}{60} )} (N) = \sum\limits_{n_1=1}^N \dfrac{1}{n_1} \sum_{n_2=1}^{n_1-1}  n_2^2\sum_{n_3=1}^{n_2-1}  n_3^2\\
&&\H_{\pi_Y(\frac{40}{3} (6x_1)^* -50(5x_1)^* +\frac{427}{6} (4x_1)^* -\frac{281}{6}(3x_1)^* + \frac{27}{2}(2x_1)^* -\frac{7}{6}x_1^* )} (N) =  \sum\limits_{n_1=1}^N n_1^2 \sum_{n_2=1}^{n_1-1} \dfrac{1}{n_2} \sum_{n_3=1}^{n_2-1}  n_3^2\\
\mbox{thus} &&\H_{\pi_Y(x_1) \stuffle \pi_Y(40(6x_1)^* - 132(5x_1)^* + 161(4x_1)^* -87(3x_1)^* +19(2x_1)^* -x_1^*)}(N) - \sum\limits_{n_1=1}^N n_1^2 \sum_{n_2=1}^{n_1-1} \dfrac{1}{n_2} \sum_{n_3=1}^{n_2-1}  n_3^2\\
&& - \sum\limits_{n_1=1}^N \dfrac{1}{n_1} \sum_{n_2=1}^{n_1-1}  n_2^2\sum_{n_3=1}^{n_2-1}  n_3^2 - \sum\limits_{n_1=1}^N n_1^2 \sum_{n_2=1}^{n_1-1}  n_2 - \sum\limits_{n_1=1}^N n_1 \sum_{n_2=1}^{n_1-1}  n_2^2 = \sum\limits_{n_1=1}^N n_1^2 \sum_{n_2=1}^{n_1-1}  n_2^2 \sum_{n_3=1}^{n_2-1} \dfrac{1}{n_3}. 
\end{eqnarray*}
\end{example}

\section{Some remarks about stuffle product and stuffle characters and their symbolic computations.}\label{stuffledef}

For the some reader's convenience, we recall here the Definitions of shuffle and stuffle products.
As regards shuffle, the alphabet $\calX$ is arbitrary and $\shuffle$ is defined by the following recursion 
(for $a,b\in\calX$ and $u,v\in\calX^*$)  
\begin{eqnarray}
u\shuffle 1_{\calX^*}=1_{\calX^*}\shuffle u=u; \ 
au\shuffle bv = a(u\shuffle bv)+b(au\shuffle v).
\end{eqnarray}
As regards stuffle, the alphabet is $Y=Y_{\N_{+}}=\{y_s\}_{s\in\N_{+}}$ and $\stuffle$ is defined by the following recursion
\begin{eqnarray}
u\stuffle 1_{Y^*}&=&1_{Y^*}\stuffle u=u,\\
y_su\stuffle y_tv&=&y_s(u\stuffle y_tv)+y_t(y_su\stuffle v) + y_{s+t}(u\stuffle v).
\end{eqnarray}
Be it for stuffle or shuffle, the noncommutative\footnote{For concatenation.} polynomials equipped with this product form an
associative commutative and unital algebra namely $(\ncp{\C}{X},\shuffle,\allowbreak1_{X^*})$ (resp. $(\ncp{\C}{Y},\stuffle,1_{Y^*})$). 

\begin{example}
As examples of characters, we have already seen
\begin{itemize}
\item $\Li_{\bullet}$ from $(\Dom^{\mathrm{loc}}(\Li_{\bullet}),\shuffle,1_{X^*})$ to $\calH(\Omega)$
\item $\H_{\bullet}$ from $(\Dom(\H_{\bullet}),\stuffle,1_{Y^*})$ to $\C^{\N}$ 
(arithmetic functions $\N\rightarrow\C$)
\end{itemize}
\end{example}
In general, a character from a $k$-algebra\footnote{Here we will use $k=\Q$ or $\C$.} $(\calA,*_1,1_{\calA})$ 
with values in $(\calB,*_2,1_{\calB})$ is none other than a morphism between the $k$-algebras $\calA$ and a
commutative algebra\footnote{In this context all algebras are associative and unital.} $\calB$.
The algebra $(\calA,*_1,1_{\calA})$ does not have to be commutative for example characters of 
$(\ncp{\C}{\calX},conc,1_{\calX^*})$ - \textit{i.e.} $\tt conc$-characters - where all proved to be of the form 
\begin{eqnarray}\label{Plane}
\biggl(\sum_{x\in\calX}\alpha_xx\biggr)^*
\end{eqnarray}
\textit{i.e.} Kleene stars of the plane \cite{CHM1,GHM1}.
They are closed under shuffle and stuffle and endowed with these laws,
they form a group. Expressions like the infinite sum within brackets in \eqref{Plane} (\textit{i.e.}
homogeneous series of degree 1) form a vector space noted $\widehat{\C.Y}$.

As a consequence, given $P=\sum\limits_{i\ge1}\alpha_iy_i$ and $Q=\sum\limits_{j\ge1}\beta_jy_j$,
we know in advance that their stuffle is a $\tt conc$-character \textit{i.e.} of the form
$(\sum\limits_{n\ge1}c_ny_n)^*$. Examining the effect of this stuffle on each letter (which suffices),
we get the identity
\begin{eqnarray}\label{WI}
\biggl(\sum_{i\ge1}\alpha_iy_i\biggr)^*\stuffle\biggl(\sum_{j\ge1}\beta_jy_j\biggr)^*=
\biggl(\sum_{i\ge1}\alpha_iy_i+\sum_{j\ge1}\beta_jy_j+\sum_{i,j\ge1}\alpha_i\beta_jy_{i+j}\biggr)^*
\end{eqnarray}  
This suggests to take an auxiliary variable, say $q$, and code ``the plane'' $\widehat{\C.Y}$,
\textit{i.e.} expressions like \eqref{Plane}, in the style of Umbral calculus by
$
\pi_Y^{\mathrm{Umbra}}:\sum_{n\ge1}\alpha_n\,q^n\longmapsto\sum_{n\ge1}\alpha_ny_n   $
which is linear and bijective\footnote{Its inverse will be naturally noted $\pi_q^{\mathrm{Umbra}}$.} from $\C_+[[q]]$ to $\widehat{\C.Y}$.
%\tcp{Very good transcription.}\\
With this coding at hand and for $S,T\in \C_+[[q]]$, identity \eqref{WI} reads 
\begin{eqnarray}\label{compo1}
(\pi_Y^{\mathrm{Umbra}}(S))^*\stuffle (\pi_Y^{\mathrm{Umbra}}(T))^*=(\pi_Y^{\mathrm{Umbra}}((1+S)(1+T)-1))^*.
\end{eqnarray}  
This shows that if one sets, for $z\in\C$ and $T\in\C_+[[x]]$,  
$
G(z)=(\pi_Y^{\mathrm{Umbra}}(e^{zT}-1))^*, 
$
we get a one-parameter stuffle group\footnote{\textit{i.e.} $G(z_1+z_2)=G(z_1)\stuffle G(z_2);G(0)=1_{Y^*}$.} such that every 
coefficient is polynomial in $z$.
Differentiating it we get 
\begin{eqnarray}\label{diffeq3}
\frac{d}{dz}(G(z))=(\pi_Y^{\mathrm{Umbra}}(T))G(z)
\end{eqnarray}
and \eqref{diffeq3} with the initial condition $G(0)=1_{Y^*}$ integrates as 
\begin{eqnarray}\label{soldiffeq3}
G(z)=\exp_{\stuffle}(z\pi_Y^{\mathrm{Umbra}}(T))
\end{eqnarray}
where the exponential map for the stuffle product is defined,
for any $P\in\ncs{\C}{Y}$ such that $\scal{P}{1_{Y^*}}=0$, is defined by 
$
\exp_{\stuffle}(P):=1_{Y^*}+\dfrac{P}{1!}+\dfrac{P\stuffle P}{2!}+\ldots+\dfrac{P^{\stuffle n}}{n!}+\ldots.
$
In particular, from \eqref{soldiffeq3}, one gets, for $k\ge1$,
$
(zy_k)^*=\exp_{\stuffle}\biggl(-\sum\limits_{n\ge1}y_{nk}\dfrac{(-z)^n}{n}\biggr).
$
\section{Conclusion}
Noncommutative symbolic calculus allows to get identities easy to check and to implement. 
With some amount of complex and functional analysis, it is possible to bridge the gap between symbolic, functional and number theoretic worlds. This was the case already for polylogarithms, harmonic sums and polyzetas. This is the project of this paper and will be pursued in forthcoming works. 
%\section*{Acknowledgements}
 %The authors would like to thank anonymous colleagues for detailed reading of the manuscript and many valuable comments by which this paper could be improved.

\end{document}

%% file: macrosB.tex
\usepackage{amsmath}
\usepackage{amsfonts}
\usepackage{amssymb}
\newcommand{\calA}{{\mathcal A}}

\newcommand{\calH}{{\mathcal H}}

\newcommand{\N}{{\mathbb N}}

\newcommand{\Z}{{\mathbb Z}}
\newcommand{\Q}{{\mathbb Q}}

\newcommand{\C}{{\mathbb C}}

\def\abs#1{|#1|}
\def\absv#1{\|#1\|}

 \def\shuffle{\mathop{_{^{\sqcup\!\sqcup}}}} 

\gdef\stuffle{\;% 
  \setlength{\unitlength}{0.0125cm}% 
  \begin{picture}(20,10)(220,580) 
  \thinlines 
  \put(220,592){\line( 0,-1){ 10}} 
  \put(220,582){\line( 1, 0){ 20}} 
  \put(240,582){\line( 0, 1){ 10}} 
  \put(230,592){\line( 0,-1){ 10}} 
  \put(225,587){\line( 1, 0){ 10}} 
  \end{picture}\; 
}

\newcommand{\Li}{\operatorname{Li}}

\def\L{\mathrm{L}}
\def\H{\mathrm{H}}

\newcommand{\poly}[2]{#1 \langle #2 \rangle}

\def\QY{\poly{\Q}{Y}}

\def\QY_0{\Q\left\langle{Y_0}\right\rangle}

 \newcommand{\calB}{{\cal B}}

\def\path{\rightsquigarrow}

\def\bv{\mid}

\def\abs#1{\bv\!#1\!\bv}

%% file: newpack3.tex
\usepackage{color}
\usepackage{ulem}
\def\scal#1#2{\langle #1\bv#2 \rangle}
\def\ncp#1#2{#1\langle #2\rangle}
\def\ncs#1#2{#1\langle \!\langle #2\rangle \!\rangle}